\documentclass{llncs}
\usepackage{color}
\usepackage{graphicx}
\usepackage{amsfonts}
\usepackage{latexsym,amssymb,amsmath}
\usepackage{mathtools}    
\usepackage[ruled]{algorithm}
\usepackage[noend]{algpseudocode}
\algdef{SE}[DOWHILE]{Do}{doWhile}{\algorithmicdo}[1]{\algorithmicwhile\ #1}

\algdef{SE}[SUBALG]{Indent}{EndIndent}{}{\algorithmicend\ }%
\algtext*{Indent}
\algtext*{EndIndent}

\def\max{{\rm max}}
\def\min{{\rm min}}
\begin{document}

\title{Rendezvous on a Line by Location-Aware Robots Despite the Presence of Byzantine Faults
\protect\footnote{This is the full version of the paper which will appear in Algosensors 2017 (The 13th International
Symposium on Algorithms and Experiments for Wireless Networks), Sep
7-8, Vienna, Austria.}}

\author{
Huda Chuangpishit\inst{1}\inst{2}
\and
Jurek Czyzowicz\inst{1}\inst{4}
\and
Evangelos Kranakis\inst{2}\inst{4}
\and
Danny Krizanc\inst{3}
}

\institute{
D\'{e}partemant d'informatique, Universit\'{e} du Qu\'{e}bec en Outaouais,  
Canada.
\and
School of Computer Science, Carleton University, Ottawa, Ontario, Canada.
\and
Department of Mathematics \& Computer Science, Wesleyan University, 
Middletown CT, USA.
\and
Research supported in part by NSERC Discovery grant.
}
\maketitle

\begin{abstract}

A set of mobile robots is placed at points of an infinite line. The robots are equipped with GPS devices and they may communicate their positions on the line to a central authority. The collection contains an unknown  subset of ``spies", i.e., byzantine robots, which are indistinguishable from the non-faulty ones.   The set of the non-faulty robots need to rendezvous in the shortest possible time in order to perform some task, while the byzantine robots may try to delay their rendezvous for as long as possible. The problem facing a central authority is to determine trajectories for all robots so as to minimize the time until the non-faulty robots have rendezvoused. The trajectories must be determined without knowledge of which robots are faulty.  Our goal is to minimize the competitive ratio between the time required to achieve the first rendezvous of the non-faulty robots and the time required for such a rendezvous to occur under the assumption that the faulty robots are known at the start. We provide a bounded competitive ratio algorithm, where the central authority is informed only of the set of initial robot positions, without knowing which ones or how many of them are faulty. When  an upper bound on the number of byzantine robots is known to the central authority, we provide algorithms with better competitive ratios. In some instances we are able to show these algorithms are optimal.

\vspace{0.5cm}

\end{abstract}


\section{Introduction}

Rendezvous is useful for cooperative  control in a distributed system, either when communication between distributed entities is restricted by range limitations or when it is required to speed up information exchanges in a distributed system. It is often presented as a consensus problem in which the agents have to agree on the meeting point and time (see~\cite{olfati2007consensus}) where by consensus we mean reaching an agreement regarding a certain quantity of interest that depends on the state of all the agents. 

In this paper we consider the following version of the rendezvous problem. A population of mobile robots is distributed at points of an infinite line. The robots are equipped with GPS devices and
are able to communicate their initial positions to a central authority. In order to perform some task, that the central authority shall assign to the robots, all of the non-faulty robots need to rendezvous (meet at the same point of the line). For this reason, the robots send to the central authority the coordinates of their positions on the line and the central authority assigns to each of them a route which eventually results in the rendezvous of all robots. A group of robots may attempt the task at any time in order to determine if
all of the non-faulty robots have been brought together. 

Unfortunately, an adversary has infected the population with ``spies" - a collection of 
byzantine faulty robots, indistinguishable from the original ones, in order to delay the performance of
the task for as long as possible. A byzantine robot may fail to report its position, report a wrong position
or it may fail to follow its assigned route. Furthermore, a faulty robot may fail to help in performing
the required task. As the central authority does not know the identity of the faulty robots it 
broadcasts travel instructions to all the robots.

We would like to define the strategy resulting in the smallest possible time of the rendezvous of
all non-faulty robots. 
Our goal is to minimize the competitive ratio between the time required to achieve this first rendezvous of the non-faulty robots and the time required for such a rendezvous to occur under the assumption that the faulty robots are known at the start. 

\subsection{Our Model}

A collection of $n$ anonymous robots travel along a Cartesian line with maximum unit speed.  Robots are equipped with GPS devices, so each of them is aware of the coordinate of its current position on the line. An unknown subset of  $f$ robots may turn out to be faulty.
At some
point in time, a task is identified that requires the coming together of all of the 
non-faulty robots at the same point
on the line and this fact is broadcast to the robots by a central authority (CA). The robots stop what they are doing
and report their positions to 
the CA. The CA computes trajectories for each of the robots and instructs them how to time their movement. 

At this point the robots follow the trajectories provided. The movement of the robots continues until such time
as all of the non-faulty robots meet for the first time and are able to perform
the task, which ends the algorithm. 
We assume the time required to attempt the
task is negligible in comparison to the time required for the robots to move between points. (As an
example, imagine that the robots have chip cards, that are used to open a
container carried by all robots. Using a secret-sharing scheme, the container is set to open only
if $n-f$ or more of the keys are valid.) A failed attempt at the task may or may not identify those
robots that are faulty (caused the attempt to fail). If identified as faulty, a robot need not continue on its trajectory.  A successful attempt at the task means that all non-faulty robots are present and
this is recognized by them and the central authority. 

As stated, we assume that the robots report their correct locations at the beginning of the algorithm. 
We note that this need only be true of the non-faulty robots as in the worst case the robots could
be anywhere and the algorithm must bring together all of them. It is
possible that faulty robots may report initial locations that are incorrect
and potentially adversely effect the lengths of the trajectories.
Of course, this may result in their receiving trajectories that they cannot complete without being
detected as faulty by the other robots. But as long as all non-faulty robots complete their trajectories
the algorithm must ensure that they meet.

The message to the CA about a robot's  position contains the robot's unique identity. We assume that the faulty robots cannot lie about their identity. Consequently, each faulty robot can send only one message about its position, otherwise it will be identified as faulty and ignored. Observe that, as the robot's identity, the CA could use the position communicated by the robots, and thus our approach could be extended to  anonymous robots. This would require some extra conditions on the model (e.g., message uniqueness), so, for simplicity, we assume that our robots have unique identifiers.

We also assume that after the initial reporting of their positions, until the reporting of success with
the task, there is no further communication between the robots themselves or the robots and the
central authority. Again, this need only be true of the non-faulty robots. Any communication by
robots during the execution of the trajectories is assumed to come from faulty robots and is
ignored. 

We note that the requirement of a central authority may be removed by allowing the robots to 
broadcast their initial positions to all other robots and each computing the same set of trajectories
using the same algorithm. 


%

A rendezvous algorithm specifies the trajectories of the robots as a function of time. We
assume the robots have sufficient memory to carry out the instructions of the rendezvous algorithm. The competitive ratio of a given  algorithm is the ratio of the time it takes the algorithm to enable rendezvous of 
all non-faulty robots divided by the time it takes the best off-line algorithm, with
knowledge of which robots are faulty, to accomplish the same. 
Note: the time of the offline algorithm equals $D/2$, where $D$ is the minimum diameter of the set of non-faulty robots. Indeed, these non-faulty robots could then meet at the mid point between the most distant ones in the set. 

We assume that the task is such that $n-f$ non-faulty robots are necessary and  sufficient to perform the task. Under this assumption, the task can be used to determine if all of the non-faulty robots are together. If a group of robots attempts the task and it succeeds, it contains all non-faulty robots. If it fails, then there exist more non-faulty robots outside the group. 

Below we present algorithms which have no knowledge of $f$ as well as others where an upper
bound on $f$ is provided. Depending upon that knowledge, different algorithms can achieve 
a better competitive ratio in different situations. We restrict our attention to the nontrivial case
where at least two robots must rendezvous, i.e.,  $f  \leq n-2$. 

\subsection{Related Work}

The mobile agent rendezvous problem has been studied extensively in many topologies 
(or domains)
and under various assumptions on system synchronicity 
and capabilities of the agents~\cite{
Czyzowicz2013,marco2006asynchronous,dessmark2003deterministic,KranakisBook} both
as a dynamic symmetry breaking problem~\cite{YuDynamic} as well as in operations research~\cite{alpern1995rendezvous} in order to understand the limitations of search theory.  
A critical distinction in the models is whether the agents must all run the same algorithm, which is generally known as the \emph{symmetric rendezvous problem} \cite{AlpernPerspective}. If agents can execute different algorithms, generally known as the \emph{asymmetric rendezvous problem}, then the problem is typically much easier, though not always trivial.

Closely related to our research is the work of  \cite{collins2011synchronous} and 
\cite{collins2010tell}. In \cite{collins2011synchronous} the authors study rendezvous of two anonymous agents, where each agent knows its own initial position in the environment, and the environment is a finite or infinite graph or a Euclidean space. They show that in the line and trees as well as in multi-dimensional Euclidean spaces and grids the agents can rendezvous in time $O(d)$, where $d$ is the distance between the initial positions of the agents. In \cite{collins2010tell} the authors study efficient rendezvous of two mobile agents moving asynchronously in the Euclidean 2d-space. Each agent has limited visibility, permitting it to see its neighborhood at unit range from its current location. Moreover, it is assumed that each agent knows its own initial position in the plane given by its coordinates. The agents, however, are not aware of each other's position. Also worth mentioning is the work of \cite{bampas2010almost} which studies the rendezvous problem of location-aware agents in the asynchronous case and whose proposed algorithm provides a route, leading to rendezvous.

The underlying domain which is traversed by the robots is a continuous curve (in our case an infinite line) and the robots may exploit a particular characteristic, e.g., different identifiers, speeds, or their initial location,  to achieve rendezvous. For example, in several papers the robots make use of the fact that they have different speeds, as in the paper~\cite{DBLP:conf/icdcn/FeinermanKKR14}, as well as in the work on probabilistic rendezvous on a cycle~\cite{DBLP:conf/icdcn/KranakisKMS15}. Rendezvous on a cycle for multiple robots with different speeds is studied in \cite{DBLP:conf/adhoc-now/HuusK15},  and rendezvous in arbitrary graphs for two robots with different speeds in~\cite{SOFSEM17}.

There is also related work on gathering a collection of identical memoryless, mobile robots in one node of an anonymous  ring whereby robots start from different nodes of the ring and operate in Look-Compute-Move cycles and have to end up in the same node~\cite{klasing2008gathering}, as well as oblivious mobile robots in the same location of the plane when the robots have limited visibility~\cite{PGNP2005}. 

Fault tolerance has been extensively studied in distributed computing, though failures were usually related to static elements of the environment, like network nodes or links (e.g., see \cite{lamport,lynch}), rather than to the mobile components. The unreliability of robots has been studied with respect to inaccurate robots' sensing or mobility devices (cf. \cite{cohen-peleg-inaccurate-sensors,izumi-unreliable-compasses,souissi2006gathering}). Problems concerning faulty robots operating in a line environment have been studied in the context of searching in \cite{CzyzowiczKKNO16} and patrolling \cite{CzyzowiczGKKKT15}. The questions of convergence or gathering involving faulty robots were investigated in \cite{agmon2006fault,BPT,cohen-peleg-convergence,DGMP,dieudonne2014gathering}.
To the best of our knowledge the rendezvous problem for location aware robots some of which may be faulty has never been considered by the research community in the past. 


\subsection{Our Results}

Here is an outline of the results of the paper. In Section~\ref{general:sec} we consider two general rendezvous algorithms for $n>2$ robots with $f \leq n-2$ faulty ones. Both algorithms assume no knowledge of the actual value of $f$ and the second algorithm stops as soon as sufficiently many robots are 
available to perform the task. The competitive ratios of these algorithms are $f+1$ and $12$, respectively.
We also prove a lower bound of $2$ on the competitive ratio for arbitrary $n>2$ and $1\leq f \leq n-2$. 
In Section \ref{sec:bounded-CR} we provide algorithms for the case where the central 
authority possesses some knowledge concerning the number of faulty robots. 
For the case where
the ratio of the number of faulty robots to the total number of robots is strictly less than 1/2 we provide an optimal algorithm and 
when this number is strictly less than 2/3 we give an algorithm that beats the general case
algorithms above unless $f$ is known to be less than 5.  
Next we provide optimal algorithms for the particular cases where $f \in \{1,2 \}$ in Section \ref{onefault:sec}. The main result here is the case of $n=4$ and $f=2$ where
we show the exact value of the competitive ratio is $1+\phi$, where $\phi$ is the golden ratio. 
We end with a discussion of open problems. 

\section{General Results}
\label{general:sec}

In this section we present a rendezvous algorithm for $n$ robots $f$ of which are faulty with a competitive ratio of at most $\min \{ f+1, 12 \}$. Neither of these algorithms require prior knowledge of $f$.
We also show that the competitive ratio of any rendezvous algorithm is at least 2. We first observe that the assumptions of our model allow us to severely restrict the potential
algorithms available to the CA. We can show the following lemma:
%
%
%
\begin{lemma}
\label{lem:Opt-move-rules}
Consider a rendezvous algorithm $A$ for $n$ robots $f$ of which are faulty with competitive ratio $\alpha$. There exists a rendezvous algorithm $B$ such that during the execution of $B$ the movement of the robots follow these rules:
\begin{itemize}
\item[(1)]
A robot does not change direction between meetings with other robots.
\item[(2)]
The robots always move at full speed. 
\end{itemize}
Moreover, the competitive ratio of $B$ is less or equal to $\alpha$. 
\end{lemma}

We assume throughout the paper that the movement of the robots in any rendezvous algorithm follows rules (1) and (2) of Lemma \ref{lem:Opt-move-rules}.

\subsection{Upper bounds}

The first rendezvous algorithm we present has a competitive ratio which is bounded above by the number of faulty robots plus one. It is interesting to note, that to obtain such competitive ratio no knowledge of the number of faulty robots is necessary. The idea of the algorithm can be summarized as follows. Consider the distances between consecutive robots on the line. The algorithm shrinks the  shortest interval (between consecutive robots) in that the two robots at its endpoints meet at its midpoint while the rest of the robots ``follow the shrinkage'' depending on their location until all non-faulty robots meet (or sufficiently many
of them in the case where the task does not require all non-faulty robots to be together to be performed). 
We prove the following theorem.

\begin{theorem}
\label{mainthm}
There is a rendezvous algorithm for $n>2$ robots at most $f$ of which are faulty whose competitive ratio is at most $f+1$, where $f \leq n-2$. 
\end{theorem}

We now describe a second general approach for rendezvous of $n$ robots, which also works for any number $f$ of faulty robots. Unlike the previous one, this algorithm  has a competitive ratio independent of 
$f$, (it equals 12). The core of our approach is the Algorithm \ref{alg:6D}, presented in \cite{collins2011synchronous}, which guarantees  rendezvous of any two robots, at initial integer positions at distance $d$ on the line, in time of at most $6d$. The idea of the algorithm is the following. Each robot gets an integer label corresponding to its initial position. The algorithm consists of a sequence of rounds, each round containing two stages. In the first round, odd-labelled robots move distance 1/2 to the right in the first stage and then distance 1 to the left in the second stage. The even-labelled robots move distance 1/2 to the left in the first stage and then distance 1 to the right in the second stage. Observe that each odd-label robot would meet its right neighbour at initial distance 1 in the first stage and its left neighbour at distance 1 in the second stage. At the end of the first round robots are in groups that from now on will travel together. 
\vspace{-0.3cm}
\begin{algorithm}[H]
\caption{Rendezvous on the infinite line}
\label{alg:6D}
\begin{algorithmic}[1]
\State{Set $\ell=\frac 12$.}
\ForAll{agents $a$}
 \State{Set $label(a)=$ position of $a$ on the line.}
\EndFor
\ForAll{$i=1, 2, 3, \ldots$}
  \ForAll{agent $a$}
\State{{\bf Stage 1.} }
\If{$odd(label(a))$}
\State{move right distance $\ell$}
\Else
\State{move left distance $\ell$.}
\EndIf
\State{{\bf Stage 2.}}
\If{$odd(label(a))$}
\State{move left distance $2\ell$.}
\Else
\State{move right distance $2\ell$.}
\EndIf
\State{$\ell=2\ell$}
\State{$label(a)=\lfloor\frac{label(a)}{2}\rfloor$}
\EndFor
\EndFor
\end{algorithmic}
\end{algorithm}
All groups are then at even distances. In round two, the configuration of such groups on the line is scaled up by the factor of two and each group of robots meet neighbouring groups at distance 2 in the two corresponding stages. The process continues inductively and after round $i$, the groups are at integer positions being multiples of $2^i$. It is possible to show that during round $i$,  in its first stage meet all robots initially placed in any interval $[(2k-1)2^i,(2k-1)2^i)$, for some integer $k$, and in its second stage meet all the robots initially placed in any interval $[(2k)2^i,(2k+2)2^i)$, for some integer $k$. Let $D$ be minimum diameter of the set of non-faulty robots required to rendezvous, and $i^*=\lceil \log_2 D \rceil$. It easy to see that all the non-faulty robots must meet in the first or the second stage of round $i^*+1$. Moreover, the total distance travelled by each robot is linear in $D$. 
In \cite{collins2011synchronous} they show the following:
\begin{theorem}[\cite{collins2011synchronous}]
\label{thm:6d}
For two agents $a_1, a_2$ starting at distance $d$ (and at integer points) on the line, Algorithm \ref{alg:6D} permits rendezvous within at most $6d$ time.
\end{theorem}
The following lemma is an immediate consequence of Theorem \ref{thm:6d}.
\begin{lemma}
\label{lem:UB-6d}
Let $a_1$ and $a_2$ be two robots on the real line with integer starting positions at distance $d$. Then the rendezvous time of $a_1$ and $a_2$ in Algorithm \ref{alg:6D} is at most $6d$. 
\end{lemma}

We now have all the required results to prove an upper bound of 12 on the competitive ratio of rendezvous of $n$ robots $f$ of which are faulty. Our approach is to approximate the initial positions of all robots by other ones which are at rational coordinates. Then the obtained configuration may be scaled up so that all initial robot positions are integers and Algorithm~\ref{alg:6D} may be applied. We show that for any $\epsilon > 0$ we can choose an approximation fine enough so that the competitive ratio does not exceed $12+ \epsilon$.  We have the following theorem.
\begin{theorem}
\label{thm:UB-12}
There exists a rendezvous algorithm for $n>2$ robots, at most $f \leq n-2$ of which are faulty, which guarantees a competitive ratio less than $12+\epsilon$, for any $\epsilon>0$.
\end{theorem}
As a corollary of Theorems \ref{mainthm} and \ref{thm:UB-12} we can state the following.
\begin{corollary}
\label{cor-min}
There is a rendezvous algorithm for $n>2$ robots at most $f \leq n-2$ of which are faulty, with competitive ratio at most $\min\{12 + \epsilon, f+1\}$, for any $\epsilon >0$.
\end{corollary}

\subsection{Lower bound}

Next we show that any rendezvous algorithm for $n$ robots, which include at least one which is faulty, must have a competitive ratio of at least 2.
\begin{theorem}
\label{mainthm1}
For any $n>2$ robots, any $1 \leq f \leq n-2$ of which are faulty, the competitive ratio of any algorithm that 
achieves rendezvous of at least $n-f$ non-faulty robots is at least $2$.
\end{theorem}
\begin{proof} (Theorem~\ref{mainthm1})
Consider the following arrangement of the robots where $n-f$ robots
are required to perform rendezvous: $\lceil \frac{f+1}{2} \rceil$ are located at position
$-1$, $n-f-1$ are located at the origin and $\lfloor \frac{f+1}{2} \rfloor$ are located at position $1$.
By Lemma \ref{lem:Opt-move-rules}, we can assume there is an optimal rendezvous algorithm
in which all robots move at speed 1 for the first 1/2 time unit. At that time, at least one of the 
robots, say $r$, starting at the origin must be at  $-1/2$ or $1/2$. Wlog, assume it is at $-1/2$. Make all of 
robots starting at the origin non-faulty, one of the robots, say $r^{\prime}$,  starting at $1$ non-faulty, and the remaining
$f$ robots faulty. In order for $n-f$ non-faulty robots to meet, $r$ and $r^{\prime}$ must meet which
requires at least another 1/2 time unit, i.e., the competitive ratio of the algorithm is at least 2. 
\qed
\end{proof}

\section{Bounded Number of Faults}
\label{sec:bounded-CR}
In the previous section we proposed algorithms, whose competitive ratio did not depend on the knowledge of the number $f$ of faulty robots. However, employing Corollary~\ref{cor-min} to get the competitive ratio which is the best between the values 12 and $f+1$ (cf. Theorems~\ref{thm:UB-12}~and~\ref{mainthm}), we need to have knowledge of an upper bound on $f$.   In this section we show, that having more precise knowledge on an upper bound on $f$ allows us to obtain algorithms with more attractive competitive ratios. More exactly, we provide upper bounds for the competitive ratio of rendezvous algorithms where the number of faulty robots is known to be bounded by a fraction of the total number of robots. 

The following theorem shows that if the majority of the robots are non-faulty then there is a rendezvous algorithm whose competitive ratio is at most 2. By Theorem \ref{mainthm1}, this is optimal.

\begin{theorem}
\label{thm:main-MTC}
Suppose that $n\geq 3$ and the number of faulty robots is $f\leq \frac{n-1}{2}$. Then there is a rendezvous algorithm with competitive ratio at most 2.
\end{theorem}
As a consequence of this result and Theorem \ref{mainthm} we get the following corollary:
\begin{corollary}
\label{cor:optimal}
If the number of faulty robots is strictly less than the number of non-faulty robots then the competitive ratio for solving the rendezvous problem is exactly 2.
\end{corollary}
%
%


In the sequel we consider the case $\frac{n-1}{2}<f<\frac 23(n-1)$
and provide an algorithm that has a better guarantee  than the general algorithm as long as our 
upper bound on
$f$ is greater than 4.

\begin{theorem}
\label{thm:UB-2n/3}
Suppose that $n\geq 3$ and there are at most $f$ faulty robots. If $f\leq \frac 23 (n-1)$ then there is a rendezvous algorithm with competitive ratio at most $5$. 
\end{theorem}

In the sequel, we present the proof of Theorem \ref{thm:UB-2n/3}. 
First note that if $n\leq 8$. Then $f\leq \frac{2}{3} (8-1)=\frac{14}{3}$, and so $f\leq 4$. Therefore by Theorem \ref{mainthm}, there is a rendezvous algorithm with competitive ratio $5$. Thus, without loss of generality we can assume that $n\geq 9$. 

\begin{lemma}
\label{lem:UB-even-3-groups}
Let $n\geq 9$ and $f< \frac 23 (n-1)$ then there is a partition of the robots into three groups $G_L$, $G_M$, and $G_R$ such that at least two of the groups $G_L$, $G_M$, and $G_R$ contain a non-faulty robot.
\end{lemma}

We are now ready to prove Theorem \ref{thm:UB-2n/3}. We present a rendezvous algorithm for the case $f<\frac 23 (n-1)$ whose competitive ratio is 5.
\begin{proof}
(Theorem \ref{thm:UB-2n/3})
Let $f< \frac 23 (n-1)$. As we discussed earlier we may assume that $n\geq 9$, as for the case $n\leq 8$ we obtain a competitive ratio of 5 by Theorem \ref{mainthm}. Therefore we can use Lemma \ref{lem:UB-even-3-groups} to split the robots into three groups $G_l$, $G_M$, and $G_R$. Consider the following rendezvous algorithm:
\vspace{-0.3cm}
\begin{algorithm}[H]
\caption{}
\label{alg:CR-5}
\begin{algorithmic}[1]
\State  {The robots broadcast their coordinates, and split into three groups as follows.}
	\Indent
	\State $G_L$ contains the $\lfloor\frac{n}{2}\rfloor-k-1$ leftmost robots
	\State $G_R$ contains the $\lceil\frac{n}{2}\rceil-k-1$ rightmost robots.
	\State $G_M$ contains the $2k+2$ middle robots.
	\EndIndent
\State {Let $A_l$ and $A_r$ be the leftmost and the rightmost robots of $G_M$, respectively. Moreover let $m_l$ and $m_r$ be the initial positions of $A_l$ and $A_r$, respectively. For the robot $A_l$, sequence all the other robots based on their distances to $A_l$ such that the robots with shorter distances appear earlier in the sequence, denote the sequence by $S_l$. Do the same for $A_r$, and let $S_r$ denotes its corresponding sequence.}
\While {the rendezvous has not occurred}
\State{ The robots in $G_L$ move at full speed to the right, and when they meet $A_l$ stick to $A_l$.}
\State{ The robots in $G_R$ move at full speed to the left, and when they meet $A_r$ stick to $A_r$.}
\State{The robots in interval $[m_l,\frac{m_l+m_r}{2})$ move towards $A_l$ and when they meet $A_l$ stick to it.} 
\State{The robots in interval $[\frac{m_l+m_r}{2}, m_r]$ move towards $A_r$, and when they meet $A_r$ stick to it.}
\State{The robot $A_l$ moves to the robot next in the sequences $S_l$, until it meets $A_r$. Then it sticks to $A_r$.}
\State{The robot $A_r$ moves to the robot next in the sequences $S_r$, until it meets $A_l$. Then it sticks to $A_l$.}
\State{When $A_l$ and $A_r$ meet they stick to each other. Then they sequence the robots based on their distances to the location of their meeting in such a way that the robots closer to the meeting point appear earlier in the sequence. Denote the sequence by $S$. The robots $A_l\cup A_r$ move to the next robot in the sequence $S$.}
\EndWhile
\end{algorithmic}
\end{algorithm}
We now analyze the competitive ratio of the above algorithm. As seen in Figure \ref{fig:3-groups}, define
\begin{itemize}
\item $B_l$: the rightmost robot in $[m_l,\frac{m_l+m_r}{2})$.
\item $B_r$: the leftmost robot in $[\frac{m_l+m_r}{2}, m_r]$.
\item $C_l$: the last robot in $G_L$ that $A_l$ meets before $A_l$ moves to visit $B_r$.
\item $C_r$: the last robot in $G_R$ that $A_r$ meets before $A_r$ moves to visit $B_l$.
\item $d_1:$ the distance between $A_l$ and $B_l$. 
\item $d_2:$ the distance between $A_r$ and $B_r$.
\item $d_3:$ the distance between $A_l$ and $C_l$. 
\item $d_4:$ the distance between $A_r$ and $C_r$.
\item $x$: the distance between $A_l$ and $A_r$. 
\end{itemize}
\begin{figure}[!htb]
\begin{center}
\includegraphics[width=10cm]{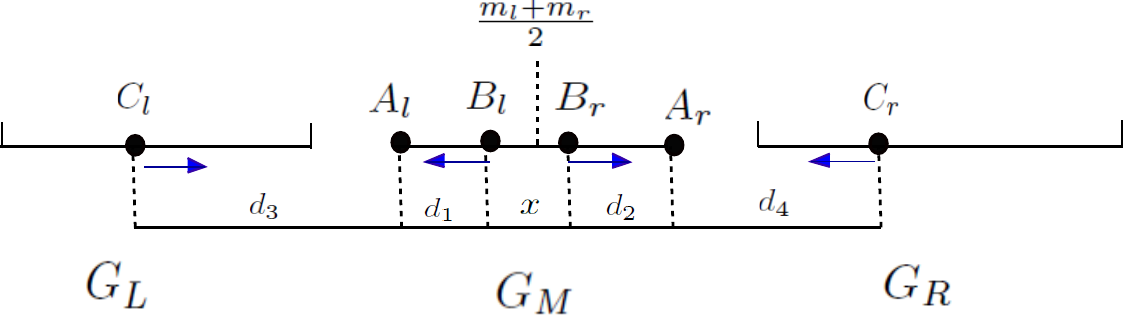}
\end{center}
\caption{}
\label{fig:3-groups}
\end{figure}
The following inequalities follow immediately.
\begin{itemize}
\item[(1)] $A_l$ meets $C_l$ before $B_r$: $d_3\leq d_1+x$.
\item[(2)] $A_r$ meets $C_r$ before $B_l$: $d_4\leq d_2+x$.
\item[(3)] Without loss of generality assume that $d_1\leq d_2$.
\end{itemize}

Let $M_l$ be the group of the robots which stick to $A_l$ before a meeting with $A_r$, see Figure \ref{fig:3-groups-proof}. More precisely $M_l$ contains $C_l$ and all the robots to the right of $C_l$, and $B_l$ and all the robots to the left of $B_l$. Similarly define $M_r$ to be the group of the robots which stick to $A_r$ before a meeting with $A_l$. Then $M_r$ contains $C_r$ and all the robots to its left, and $B_r$ and all the robots to its right. 
\begin{figure}[!htb]
\begin{center}
\includegraphics[width=10cm]{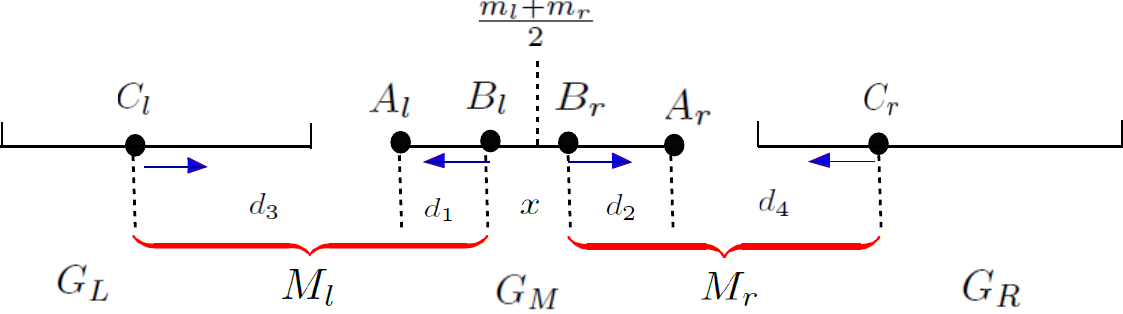}
\end{center}
\caption{}
\label{fig:3-groups-proof}
\end{figure}

Consider the following three cases:

\vspace{.3cm}
\noindent{\bf Case 1. The rendezvous occurs among $M_l$ or $M_r$:} Without loss of generality assume that the rendezvous occurs among $M_L$. This implies that the non-faulty robots belong to $G_L$ and the interval $[m_l,\frac{m_l+m_r}{2})$. Let $a_l$ be the leftmost non-faulty robot of $G_L$ and $a_r$ be the rightmost non-faulty robot of $[m_l,\frac{m_l+m_r}{2})$. The rendezvous of Algorithm \ref{alg:CR-5} occurs when $A_l$ meets both $a_r$ and $a_l$. Suppose that $\delta_1$ and $\delta_2$ are the distances between $A_l, a_r$ and $A_l,a_l$ respectively. Since $A_l$ moves towards the closest robots then the rendezvous occurs at the time at most $\frac 32\max\{\delta_1,\delta_2\}$. Moreover $\max\{\delta_1,\delta_2\}$ is bounded above by the diameter of non-faulty robots. Therefore in this case the competitive ratio is at most 3.

\vspace{.3cm}
\noindent{\bf Case 2. The rendezvous occurs at the time of the meeting of $M_l$ and $M_r$:} 
The meeting of $M_l$ and $M_r$ occurs when the robots $C_l$ and $C_r$ meet. The robot $C_l$ moves to the right and the robot $C_r$ moves to the left, and thus their meeting occurs at time $$\frac{d_1+d_2+d_3+d_4+x}{2}.$$ By Inequalities (1), (2) and (3) we have
\begin{align*}
d_1+d_2+d_3+d_4+x &\leq 2d_1+2d_2+3x \\
&=4d_1+2x_1+3x\\
&=4(d_1+x)-x+2x_1
\end{align*}
By Lemma~\ref{lem:UB-even-3-groups} 
we know that  at least two of $G_L$, $G_M$, and $G_R$ contain non-faulty robots. This implies that the diameter of the non-faulty robots, $D$, is at least $\min\{d_1+x+z,d_2+x+y\}$. By Inequality (3) we have that $D\geq d_1+x$. Therefore
\begin{align*}
CR &\leq \frac{4(d_1+x)-x+2x_1}{d_1+x}\\
&=4+\frac{x_1-x_2}{d_1+x}\leq 5
\end{align*}  
\noindent{\bf Case 3. The rendezvous occurs after the meeting of $M_l$ and $M_r$:} This case occurs if there are non-faulty robots either to the left of $C_l$ or to the right of $C_r$. First assume that there are non-faulty robots both to the left of $C_l$ and to the right of $C_r$. Then the rendezvous occurs in at most two times the diameter of non-faulty robots. Therefore the competitive ratio is at most 2. Now suppose without loss of generality that there are only non-faulty robots to the right of $C_r$, and the rightmost non-faulty robot is $R$ at distance $\delta$ from $C_r$. Then at the time of the meeting of $M_l$ and $M_r$ the distance between the robots in $M_l\cup M_r$ and $R$ is $\delta$. So it takes at most $\frac{3}{2}\delta$ for $M_l\cup M_r$ to meet $R$. Therefore the rendezvous occurs at time $$\frac{d_1+d_2+d_3+d_4+x}{2}+\frac{3\delta}{2}$$ while the optimal time is at least $\frac{d_4+\delta}{2}$.

Note that $A_r$ meets $A_l$ before $R$, and thus $d_1+d_2+x\leq d_4+\delta$. Moreover $d_3\leq d_2+x$. Therefore
$$
\frac{d_1+d_2+d_3+d_4+x}{2}+\frac{3\delta}{2}\leq \frac{3(d_4+\delta)}{2}
$$
This implies that the competitive ratio in this case is at most 3.
\qed
\end{proof}
%
%
\section{Optimal Rendezvous Algorithms for at Most Two Faulty Robots}
\label{onefault:sec}
This section is dedicated to the study of optimal rendezvous algorithms when the number of faulty robots is small, i.e., for $f\in\{1,2\}$. 

The next theorem yields the competitive ratio for $f=1$ fault and is an immediate consequence of Theorems~\ref{mainthm}~and~\ref{mainthm1}.
\begin{theorem}
\label{thm:one-Faulty}
For $n>2$ robots with $f=1$ faulty, the competitive ratio of the algorithm which shrinks the shortest interval 
is 2, and this is optimal.
\end{theorem}

It remains to consider the competitive ratio for $f=2$ faulty robots. By Corollary \ref{cor:optimal}, the competitive ratio of the rendezvous problem for $n \geq 5$ robots with two faulty is exactly 2. Therefore the only unknown case concerning two faulty robots is when $n=4$. In this section we present a rendezvous algorithm with optimal competitive ratio $1+\phi$, where $\phi = \frac{1+\sqrt{5}}{2}$ is the golden ratio.  We summarize the main result in the following two theorems.
For the lower bound we prove:
\begin{theorem}
\label{thm:4-robots-thm}
Consider four robots exactly two of which are faulty. No rendezvous algorithm can have competitive ratio less than $1 + \phi$, where $\phi$ is the golden ratio.
\end{theorem}
The proof of Theorem~\ref{thm:4-robots-thm} is based on an exhaustive analysis and considers the competitive ratio of any potential algorithm solving the rendezvous problem for the four robots.
For the upper bound we prove:
\begin{theorem}
\label{thm:4-robots-UB}
Consider four robots exactly two of which are faulty. There is a rendezvous algorithm for four robots two of which are faulty with competitive ratio at most $1+\phi$, where $\phi$ is the golden ratio.
\end{theorem}
The proof of Theorem~\ref{thm:4-robots-UB} is a continuation of the proof of Theorem~\ref{thm:4-robots-thm} leading to a specific algorithm whose competitive ratio is optimal for the rendezvous problem considered.


\section{Conclusion}

In this paper we considered the rendezvous problem for $n>2$ robots on a line with $1 \leq f \leq n-2$ among them byzantine faulty.  The robots were equipped with GPS devices and they could communicate their positions to a central authority.  We designed several rendezvous algorithms and considered their competitive ratio depending on the knowledge the central authority has about the number of faulty robots. An interesting question remaining might be to improve the competitive of the algorithms presented.
Another question concerns the model presented here which ignores any communication beyond the broadcasting of the initial positions of the robots. It might be of interest to consider algorithms in a ``richer'' communication model where the robots may broadcast information as they follow their trajectories. For example, one could consider a model where the faulty robots may crash and non-faulty robots may report not meeting them when expected. 

\newpage

\bibliographystyle{plain}
\bibliography{refs,refs1}

\newpage
\appendix

\section{Proof of Lemma \ref{lem:Opt-move-rules}}

\begin{proof}(Lemma \ref{lem:Opt-move-rules})
Consider the algorithm $A$ with competitive ratio $\alpha$. If the movement of the robots in $A$ follows (1) and (2) then there is nothing left to prove. So suppose that $r$ is a robot whose movement does not follow (1), i.e., at some point of the execution of $A$ the robot $r$ changes its direction without meeting any robots. Consider the last time that $r$ changes its direction between the meetings. So let $y$ be the last point where $r$ changes its direction and there is no robot located at $y$. We can assume without loss of generality that $r$ was moving to the right when it meets $y$ and then $r$ starts moving to the left to meet another robot $r'$. Let $(x,y]$ be the largest interval for which $r$ moves without meeting any robot until it reaches $y$. Then $r$ moves to the left of $y$ to meet the robot $r'$. But if $r$ moves to the left from the point $x$ it will meet $r'$, at least $y-x$ sooner. So clearly if we cut off the movement of $r$ to the right during $(x,y]$, and instead let $r$ move to the left from the point $x$ then the rendezvous time of the modified algorithm is less or equal to that
of $A$. Therefore for all robots which change direction before a meeting we modify their trajectories as discussed for $r$. Let $B_1$ be the algorithm with the modified trajectories. Note that the worst
case rendezvous time of the algorithm $B_1$ is less than or equal to that of algorithm $A$.

The movement of all robots in $B_1$ follows rule (1). Now suppose that there is a robot $r$ whose movement in the algorithm $B_1$ does not follow (2), i.e., at some point of the execution of $B_1$ the robot $r$ either slows down or stops moving. First suppose that $r$ slows down at the point $x$. Since the rendezvous has not occurred yet we know that $r$ moves to meet another robot $r'$. It is clear that if $r$ moves at its full speed the meeting of $r$ and $r'$ occurs sooner. Therefore moving at full speed does not increase the rendezvous time of the algorithm $B_1$. Define $B_2$ to be the algorithm which copies all the steps of the algorithm $B_1$ with robots always moving at the full speed.

Now suppose that there is a robot $r$ which stops at point $x$ during the execution of $B_2$ before the rendezvous occurs. Note that if $r$ starts moving after a stop then clearly moving at full speed to the same direction does not increase the rendezvous time. 

Let $b_1, \ldots, b_m$ be the sequence of robots in the increasing order of their meeting time with $r$, after $r$ stops at $x$. Instead of waiting at $x$, we ask $r$ to move towards the next robot in the sequence $b_1, \ldots, b_m$. We now prove that this does not increase the meeting time of $r$ and $b_i$, $1\leq i\leq m$. Note that at some point the robot $b_i$ starts moving towards $r$. 

Let $d_i$ be the distance between $b_i$ and $r$ at the time when $b_i$ starts moving towards $r$. Moreover let $\delta_{i+1}$ be the time of the meeting of $b_{i+1}$ and $r$ from the moment that $b_i$ starts moving towards $r$. Since $r$ meets $b_i$ before $b_{i+1}$ then $d_i\leq \delta_{i+1}$. If $r$ moves towards $b_i$ they meet after $\frac{d_i}{2}$, and $r$ can move back to $x$ by moving another $\frac{d_i}{2}$. Therefore $r$ spends $d_i\leq \delta_{i+1}$ to visit $b_i$ and return to its position $x$. This implies that if $r$ moves towards the next robot in the sequence $b_1, \ldots, b_m$ the rendezvous time of Algorithm $B_2$ does not increase.    

We apply this to all robots such as $r$, and denote the modified algorithm by $B$. The rendezvous times of $B$ are less than or equal to those of $B_2$. 
Therefore $B$ is a rendezvous algorithm with competitive ratio less or equal to $\alpha$  in which the movement of all robots follow (1) and (2). 
\qed
\end{proof}

\section{Proof of Theorem~\ref{mainthm}}

\begin{proof} (Theorem~\ref{mainthm})
Consider the following algorithm which enables rendezvous of the robots by shrinking the smallest interval between robots until  $n-f$ non-faulty robots meet at the same point.

\vspace{-0.3cm}
\begin{algorithm}[H]
\caption{SSI (ShrinkShortestInterval)}
\label{alg:ssi}
\begin{algorithmic}[1]
\State{The robots broadcast their coordinates.}
\While{insufficient number of non-faulty robots have met}
\State{Compute the smallest distance between any two consecutive robots.}
\State{Select a(ny) pair of robots, whose distance is minimum; say these robots are $r, r'$ at distance $d$ and let $r$ lie to the left of robot $r'$.}
\State{Robots $r, r'$ move towards each other for distance $d/2$ until they meet. Moreover, all robots to the left of $r$ follow $r$ by shifting distance $d/2$ to the right, while robots to the right of $r'$ follow $r'$ by shifting distance $d/2$ to the left.}
\State{The two robots merge the groups they belong to.
} 
\EndWhile
\end{algorithmic}
\end{algorithm}
Notice that when robots meet they can attempt the task in order to determine if
all of the non-faulty robots have met. Further, observe that the algorithm determines
trajectories for each of the robots (depending only upon the inter-robot distances) which the robots
follow until at least $n-f$ non-faulty robots have met. In the worst case they all converge on a single
point. 

Next we analyze the competitive ratio of the algorithm SSI.
Recall that at the beginning of the algorithm all the robots announce their coordinates. 
Let $x < y$ be the locations of the leftmost and rightmost non-faulty robots.
%

Observe that all the robots to the left of position $x$ and to the right of position $y$ are 
faulty. Let $d_1 \leq d_2 \leq \cdots \leq d_{n-1}$ be the multiset of inter-robot distances in non-descending order.  Consider the inter-robot distances between all the robots whose starting positions are between the starting positions of $x$ and $y$ and denote them by $\delta_{\min} := \delta_1 \leq \delta_2 \leq \cdots  \leq \delta_{\max} := \delta_t$, where $\delta_{\min}, \delta_{\max}$ are the minimum and maximum inter-robot distances between $x$ and $y$.

Observe further that once the interval $\delta_{\max}$ is collapsed all of the non-faulty robots are together in one spot and rendezvous has occurred.  This implies that if $d_s := \delta_t$ in the multi-set of these distances then we have that the competitive ratio (of the algorithm), denoted by $CR$, satisfies the following equalities.
\begin{align*}
CR 
&= \frac{\frac{d_1 + d_2 + \cdots + d_s}{2}}{\frac{\delta_1 + \delta_2+ \cdots +\delta_t}{2}} 
= \frac{d_1 + d_2 + \cdots + d_s}{\delta_1 + \delta_2+ \cdots +\delta_t}\\
&= \frac{\delta_1 + \delta_2+ \cdots +\delta_t + \sum_{d_u \not\in \{ \delta_1 , \delta_2, \ldots ,\delta_t\}} d_u }{\delta_1 + \delta_2+ \cdots +\delta_t} \\
&= 1 + \frac{\sum_{d_u \not\in \{ \delta_1 , \delta_2, \ldots ,\delta_t\}} d_u }{\delta_1 + \delta_2+ \cdots +\delta_t}  .
\end{align*}

Consider a distance $d_u \not\in \{ \delta_1 , \delta_2, \ldots ,\delta_t\}$ occurring in the sum in the numerator of the last fraction above. Observe that unless $d_u \leq \delta_{\max}$, the SSI algorithm above will choose to shrink one of the intervals in the set $\{ \delta_i \}$ of distances of robots between $x$ and $y$. 

Observe that there are at most $f$  robots initially located outside the interval $[x,y]$. From this we conclude that
\begin{align}
CR & \label{distances0}
\leq 1 + \frac{f \delta_{\max} }{\delta_1 + \delta_2+ \cdots +\delta_t} 
\\
& \label{distances1}
\leq 1+ f .
\end{align}
This completes the proof of Theorem~\ref{mainthm}.
\qed
\end{proof}

{\it Remark:} The bound of $f+1$ is tight in that it easy to construct an example where the ratio
is as close to this bound as desired. For example, consider $f$ faulty robots at positions 1 through $f$ on the line,
one non-faulty robot at $f+1$ and the remaining $n-f-1$ robots at position $f+2+\epsilon$,
for any $\epsilon > 0$. Following
the algorithm, rendezvous does not occur until time $f/2 + \frac{1+\epsilon}{2}$ but could have
occurred at time $\frac{1+ \epsilon}{2}$ with a competitive ratio of $1 + \frac{f}{1+\epsilon}$.

\section{Proof of Theorem \ref{thm:UB-12}}

\begin{proof} (Theorem \ref{thm:UB-12})
Let $a_1, \ldots, a_n$ be the positions of the robots, and $d$ be the minimal distance between the robots, i.e.,
$d=\min_{1\leq i<j\leq n}|a_i-a_j|.$
For any initial position of the robot $r_i$ choose a rational number $\frac{p_i}{q_i}$ in the interval $(a_i-\frac{\epsilon d}{24}, a_i+\frac{\epsilon d}{24})$. Scale the real line as follows. Let $q= {\rm lcm} (q_1,\ldots, q_n)$ be the least common multiple of the denominators. Map $x\in \mathbb{R}$ to $qx$. Then every robot $r_i$ has an integer position denoted by $a_i^* = qp_i/q_i$. Note that by the choice of the rational numbers $\frac{p_i}{q_i}$, the order of $a_i^*$s on the line is the same as the order of $a_i$s.
Now run Algorithm \ref{alg:6D} for the $n$ robots at their new positions $a_1^*,\ldots, a_n^*$. Let $D$ be the diameter of the non-faulty robots with original positions $a_1, \ldots, a_n$, i.e. $D=a_M-a_m$ where $a_M$ is the greatest position of a non-faulty robots and $a_m$ is the smallest one.  

Similarly, let $D^*$ be the diameter of the non-faulty robots with integer positions $a_1^*, \ldots, a_n^*$, i.e. $D^*=a_M^*-a_m^*$ where $a_M^*$ is the greatest integer position of the non-faulty robots and $a_m^*$ is the smallest one. Elementary calculations using the definitions of $D$ and $D^*$ above show that $D^*  \leq q( D+\frac{\epsilon d}{12})$. 

By Lemma \ref{lem:UB-6d} we have that the robots at (integer) positions $a_M^*$ and $a_m^*$ meet in time smaller than $6D^*\leq 6q (D+\frac{\epsilon d}{12})$. Moreover by Lemma \ref{lem:UB-6d} all the robots in the interval $[a_m^*, a_M^*]$ meet by that time.

It follows that the rendezvous time of the robots in their original positions is bounded from above by  $(6+\frac{\epsilon d}{2})D$ while the optimal offline time is equal to $D/2$. Hence for any $\epsilon>0$ the competitive ratio is less than $\frac{(6+\epsilon d/2)D}{D/2}=12+\epsilon d /D \leq 12 + \epsilon$.
\qed
\end{proof}

\section{Proof of Theorem \ref{thm:main-MTC}}

\begin{proof}
(Theorem \ref{thm:main-MTC})
Since $f\leq \frac{n-1}{2}$, $n\geq 2f+1$. Consider the following rendezvous algorithm. 

\vspace{-0.3cm}
\begin{algorithm}[H]
\caption{ MTC (Move Towards the Center)}
\label{alg:mtc}
\begin{algorithmic}[1]
\State{The robots broadcast their coordinates.}
\If{$n\geq 2f+2$}
\State{The robots split into two groups $G_L$ and $G_R$ consisting of the $\lfloor n/2 \rfloor$ leftmost and $\lceil n/2 \rceil$ rightmost robots, respectively.}
\State{The robots in $G_L$ move right with speed $1$ and the robots in
$G_R$ move left with speed $1$.}
\State{The first two robots (one from $G_L$, one from $G_R$ rendezvous and make a group $A_0$. Go to 10.}
\EndIf
\If{$n=2f+1$}
\State{The robots split into two groups $G_L$ and $G_R$ consisting of the $\lfloor n/2 \rfloor$ leftmost and $\lfloor n/2 \rfloor$ rightmost robots.}
\State{The robots in $G_L$ move right with speed $1$ and the robots in
$G_R$ move left with speed $1$.}
\State{Let $A_0$ be the middle robot.}
\EndIf
\State{The robot $A_0$ computes its distance with all the robots to its right and left, and arranges them in the increasing order of their distances to itself as $b_1, \ldots, b_m$. }
\State{Let $i=0$.}
\While{$A_i$ does not contain a sufficient number of non-faulty robots} 
\State{The group $A_i$ moves towards the robot $b_{i+1}$ to form the group $A_{i+1}$.}
\State{$i=i+1$}
\EndWhile
\end{algorithmic}
\end{algorithm}
Recall that the robots in the groups $A_i$ can use the task to determine if $n-f$ non-faulty 
robots have met. Now we analyze the competitive ratio of the algorithm MTC.

We consider two cases:

{\it Case 1: $n>2f+1$.}
In this case both $G_L$ and $G_R$ are guaranteed to contain at least one non-faulty robot. Also observe
that neither $G_L$ nor $G_R$ contains $n-f$ robots, i.e., for rendezvous to occur at least one robot
from each of $G_L$ and $G_R$ must be involved. Let $r_L$ (respectively, $r_R$) be the rightmost
(leftmost) robot in $G_L$ ($G_R$) which is contained in an minimum length interval containing $n-f$
non-faulty robots. Let $q_L$ (respectively, $q_R$) be the rightmost (respectively, leftmost) robot in
$G_L$ (respectively, $G_R$). Let $x$ be the distance from $r_L$ to $q_L$, $y$ the distance from
$q_L$ to $q_R$ and $z$ the distance from $q_R$ to $r_R$. Then an offline algorithm can
perform rendezvous in $(x+y+z)/2$. Observe that $A_0$ is formed at time $y/2$, and that $r_L$
(respectively, $r_R$) joins
the central group at time less or equal to $x$ (respectively, $y$) after that. Therefore rendezvous occurs at time at most $y/2 + \max(x,z)$,
i.e., the competitive ratio is at most 2. 

{\it Case 2: $n=2f+1$.} Here there are two subcases:

{\it Case 2a: The robot in $A_0$ is faulty.}
In this case, there must be a non-faulty robot 
that is in $G_L$, as well as one in $G_R$. At this point we may argue as in Case 1 with 
$y=0$. 

{\it Case 2b: The robot in $A_0$ is non-faulty.} In this case, either the minimum length interval 
containing $n-f$ robots has at least one robot in each of $G_L$ and $G_R$, in which case 
we argue as above, or it doesn't. In this second case, the interval ends with the robot in $A_0$
with all of the non-faulty robots in either $G_L$ or $G_R$. Say its $G_L$. In this case, the time until rendezvous
is the time for the leftmost robot in $G_L$ to join the central group, which is at most the distance
between that robot and $A_0$ which is twice the optimal time for rendezvous, i.e., the
competitive ratio is at most 2. 

\qed
\end{proof}

\section{Proof of Lemma~\ref{lem:UB-even-3-groups}}

\begin{proof} (Lemma~\ref{lem:UB-even-3-groups})
We know that $f< \frac 23 (n-1)=\frac n2+\frac n6-\frac 23$. Let $k=\lfloor \frac n6-\frac 23\rfloor$. Since $n\geq 9$ we have $k\geq 0$. Split the robots as follows:
\begin{itemize}
\item $G_L$: the group of $\lfloor \frac{n}{2}\rfloor-k-1$ leftmost robots.
\item $G_R$: the group of $\lceil\frac{n}{2}\rceil-k-1$ rightmost robots.
\item $G_M$: the group of the remaining $2k+2$ robots in the middle.
\end{itemize}
Next assume to the contrary that:
\begin{itemize}
\item
$G_L$ and $G_R$ contain all faulty robots: Then
$$f\geq \left(\left\lfloor\frac{n}{2}\right\rfloor-k-1\right)+\left( \left\lceil\frac{n}{2}\right\rceil-k-1\right)=n-2k-2.$$
Since $f< \frac 23 (n-1)$ then $n-2k-2< \frac{2n}{3}-\frac 23$. This implies that $k>\frac n6-\frac23$, which is a contradiction.
\item
$G_M$ and $G_L$ contain all faulty robots: Then
$$f\geq \left(\left\lfloor\frac{n}{2}\right\rfloor-k-1\right)+(2k+2)=\left\lfloor\frac{n}{2}\right\rfloor+k+1.$$
Since $f\leq\frac{n}{2}+k$ then $\lfloor\frac{n}{2}\rfloor+k+1\leq \frac{n}{2}+k$. This implies that $k+1\leq k+\frac n2-\lfloor\frac{n}{2}\rfloor$. But this is a contradiction since $\frac n2-\lfloor\frac{n}{2}\rfloor<1$. 
\item
$G_M$ and $G_R$ contain all faulty robots: Then
$$f\geq \left(\left\lceil\frac{n}{2}\right\rceil-k-1\right)+(2k+2)=\left\lceil\frac{n}{2}\right\rceil+k+1.$$
Since $f\leq\frac{n}{2}+k$ then $\lceil\frac{n}{2}\rceil+k+1\leq f\leq\frac{n}{2}+k$. This implies that $k+1\leq k+\frac n2- \lceil\frac{n}{2}\rceil$. But this is a contradiction since $\frac n2-\lfloor\frac{n}{2}\rfloor <1$. 
\end{itemize}
Therefore at least two of $G_L, G_M$, and $G_R$ contain non-faulty robots.
\qed 
\end{proof}

\section{Proof of Theorem \ref{thm:4-robots-thm}}

The four robots are labeled $a,b,c,d$ and are drawn left to right in this order as depicted in Figure~\ref{fig:cases-line}. Without loss of generality we may assume that robot $a$ occupies position $0$ and robot $d$ position $1$.
\begin{figure}[!htb]
\begin{center}
\includegraphics[width=8cm]{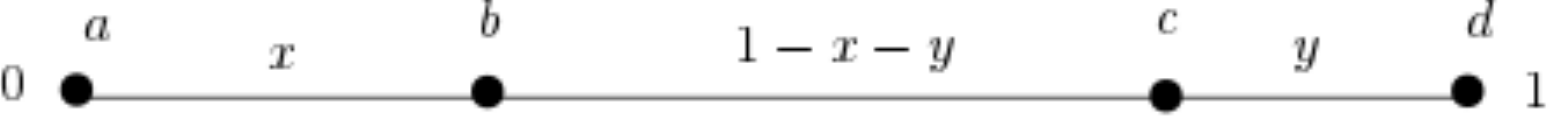}
\end{center}
\caption{The line segment $[0,1]$ and four robots $a, b, c, d$ at respective consecutive distances $x, 1-x-y, y$. Without loss of generality we may assume that $y \leq x, 1-x$.}
\label{fig:cases-line}
\end{figure}
Two of the robots are faulty but it is not known exactly which two. The robots start moving at the same time and move with maximum speed $1$. The corresponding initial distances of the robots are as follows: $a, b$ are at distance $x$, $b, c$ at distance $1-x-y$, and $c,d$ at distance $y$. Without loss of generality we may assume $0 \leq y \leq x \leq 1$ and $y \leq 1-x$. 

The rest of this section is devoted to the proof of Theorem \ref{thm:4-robots-thm}.
We want to consider any potential algorithm solving the rendezvous problem for the four robots as depicted in Figure~\ref{fig:cases-line}. 
It turns out that the worst case competitive ratio occurs when $x=\phi y, 1-x-y = x$, where $\phi = \frac{1+\sqrt{5}}{2}$ is the golden ratio. Therefore, $\phi y + \phi y + y  = 1$ and hence $y = \frac{1}{1+2 \phi} = \frac{1}{2 + \sqrt{5}}$ and $x = \phi y = \frac{1+\sqrt{5} }{4 + 2\sqrt{5}}$. 

As stated in Lemma~\ref{lem:Opt-move-rules} we need only consider rendezvous algorithms in which the movement of the robots obey rules (1) and (2) of Lemma \ref{lem:Opt-move-rules}. Figure~\ref{fig:cases-tree} depicts the order in which the rendezvous of the four robots may occur in six possible cases.
\begin{figure}[!htb]
\begin{center}
\includegraphics[width=8cm]{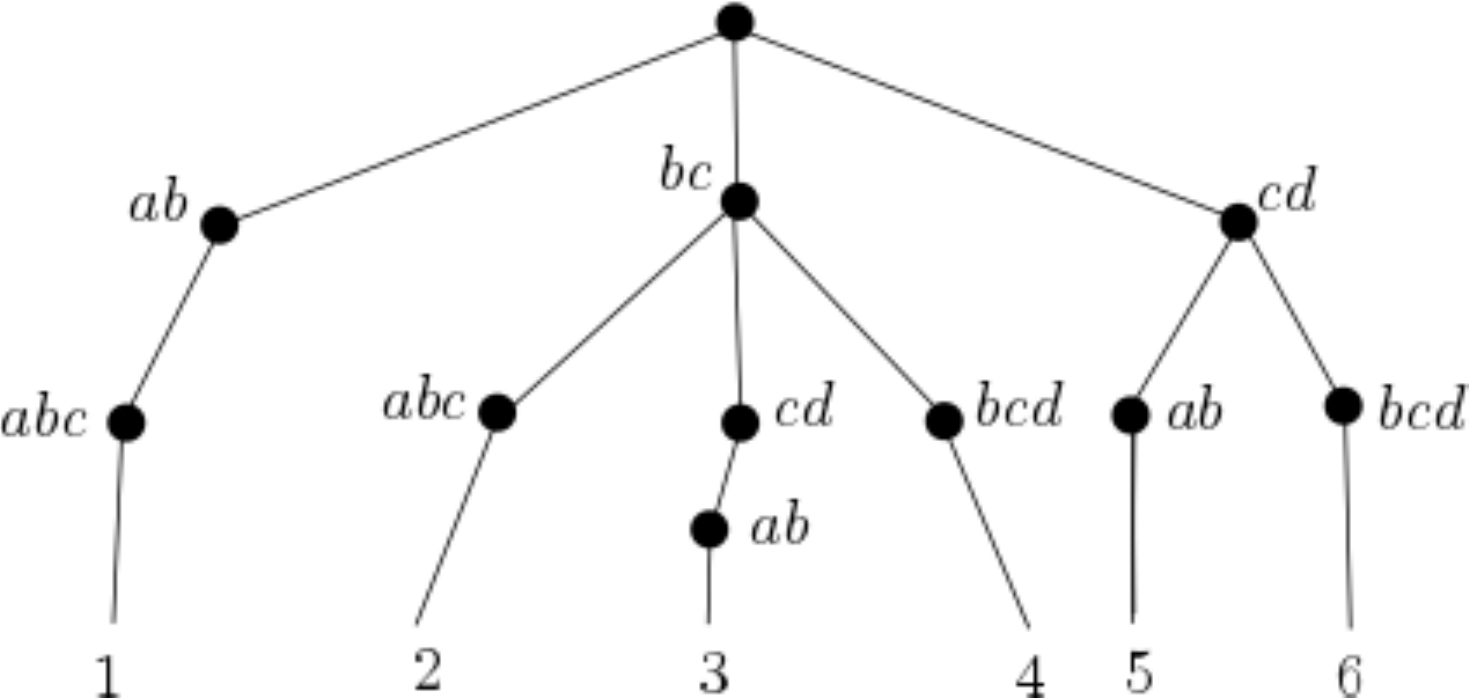}
\end{center}
\caption{Tree representing the six possible cases of rendezvous for four robots $a, b, c, d$ on a line. For each leaf of the tree the unique path of the tree depicts the sequence of consecutive rendezvous taking place until all four robots meet. Each number on a leaf represents the corresponding sequence of rendezvous events from the root to this leaf.}
\label{fig:cases-tree}
\end{figure}
For example, either of the rendezvous $ab, bc, cd$ may occur first. If the rendezvous $bc$ occurs first then it may followed by any of the three possible rendezvous $abc, cd, bcd$, etc.  The other cases are similar. From Figure~\ref{fig:cases-tree} it is easily seen that there are six cases which describe any possible rendezvous algorithm of the four robots. 
Next we need to analyze the six cases and determine upper and lower bounds on the rendezvous times.

As depicted in Table~\ref{tbl:cases}, there are six functions $f_i (x,y)$, for $i=1,\ldots , 6$, arising, and each of which represents the competitive ratio of rendezvous time for the corresponding case. The first column depicts the Case being considered, the second column the pairs of the non-faulty robots which rendezvous in optimal time, the third column the pairs of the non-faulty robots with non-optimal rendezvous, the fourth column the sequence of the distances moved, and the last column the competitive ratio of the corresponding case.
\begin{table}[htp]
\begin{center}
\begin{tabular}{| c | c | c | c | c |}
\hline
Case & Optimal & Non-Optimal & Distance Moved & Competitive Ratio \\
\hline
1 & $ad, ab, ac$ & $bc, cd, bd$ & $\frac x2, \frac{1-x-y}2, \frac y2$ & $f_1 (x, y) := \max \left\{ \frac{1-y}{1-x-y}, \frac{1}{y} \right\}$ \\
\hline
2 & $ad, bc, ac$ & $ab, bd, cd$ & $\frac{1-x-y}2, \frac x2, \frac y2$  & $f_2 (x, y) := \max \left\{ \frac{1-y}{x}, \frac{1}{y} \right\}$ \\
\hline
3 & $ad, bc$ & $ab, ac, bd, cd$ & $\frac{1-x-y}2, \frac y2, \frac{x-y}2 ,\frac y2$ & $f_3 (x, y) := \max \left\{ \frac{1}{1-x}, \frac{1-x}{y} \right\}$ \\
\hline
4 & $ad, bc, bd$ & $ab, ac, cd$ & $\frac{1-x-y}2, \frac y2, \frac x2$ & $f_4 (x, y) := \max \left\{ \frac{1}{x}, \frac{1-x}{y} \right\}$ \\
\hline
5 & $ad, ab, cd$ & $ac, bd, bc$ & $\frac y2, \frac {x-y}2, \frac {1-x}2$ & $f_5 (x, y) := \frac{1}{1-x-y}$ \\
\hline
6 & $ad, cd, bd$ & $ab, ac, bc$ & $\frac y2, \frac{1-x-y}2, \frac x2$ & $f_6 (x, y) := \max \left\{ \frac{1}{x}, \frac{1-x}{1-x-y} \right\}$ \\
\hline
\end{tabular}
\end{center}
\caption{The six rendezvous cases enumerated in column 1. The second column depicts the pairs of non-faulty robots with optimal rendezvous, the third column the pairs of non-faulty robots with non-optimal rendezvous, the fourth column the times the different rendezvous between robots occur, 
and the fifth (rightmost) column the competitive ratio of the rendezvous pairs of robots.
}
\label{tbl:cases}
\end{table}%
In sequel we explain the details of Table~\ref{tbl:cases}. 
First note that for the four robots $a, b, c$, and $d$ there are six possible combinations of two non-faulty robots: $ab, ac, ad, bc, bd$, and $cd$. Also recall that the movement of all the robots follow rules (1) and (2) of Lemma \ref{lem:Opt-move-rules}, {\em i.e}, they do not change direction between the meetings and they move at the full speed. 

\vspace{.2cm}
\noindent{\bf Case 1 Rendezvous:} As depicted in Figure \ref{fig:case1}, robots $a$ and $b$ meet first, then $ab$ meets $c$ and at the end $abc$ and $d$ rendezvous. 
\begin{figure}[ht]
\begin{center}
\includegraphics[width=8cm]{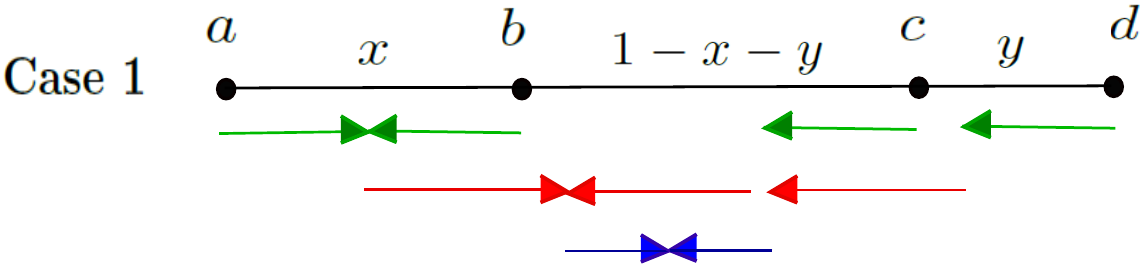}
\end{center}
\caption{Case 1 of Table \ref{tbl:cases}.}
\label{fig:case1}
\end{figure}
\begin{itemize}
\item[(1)] 
{\bf The pairs of non-faulty robots with optimal and non-optimal rendezvous time (columns 2 and 3 of Table \ref{tbl:cases}):} If $ac$, $ab$ or $ad$ are non-faulty then the rendezvous of Case 1 is optimal. This is due to the fact that $a$ moves in the right direction and $c$, $b$, and $d$  move in the left direction with their full speed. Therefore, if $ac$, $ab$ or $ad$ are non-faulty then the competitive ratio is one. For the cases when one of the pairs $bc$, $bd$, and $cd$ are non-faulty the rendezvous of Case 1 is not optimal.
\item[(2)] 
{\bf The distance moved (column 4 of Table \ref{tbl:cases}):} 
\begin{itemize}
\item
Rendezvous of $a$ and $b$: The distance between $a$ and $b$ is $x$. So $a$ and $b$ move towards each other for $\frac{x}{2}$. Also, $c$ and $d$ move to the left for $\frac x2$.
\item
Rendezvous of $ab$ and $c$: The the distance between $ab$ and $c$ is $1-x-y$, and so $ab$ and $c$ move $\frac{1-x-y}{2}$ towards each other to meet. Also $d$ moves $\frac{1-x-y}{2}$ to the left.
\item
Rendezvous of $abc$ and $d$: the distance between $abc$ and $d$ is $y$. So to finish the rendezvous $abc$ moves $\frac{y}{2}$ to the right and $d$ moves $\frac{y}{2}$ to the left. 
\end{itemize}
\item[(3)]
{\bf The competitive ratio (column 5 of Table \ref{tbl:cases}):} let one of the pairs $bc$, $bd$, and $cd$ be non-faulty.
\begin{itemize}
\item
$bc$ is non-faulty: then the rendezvous of Case 1 is complete when $ab$ and $c$ meet. Therefore the robots $b$ and $c$ move $\frac{x}{2}+\frac{1-x-y}{2}$ to rendezvous while in the optimal algorithm $b$ and $c$ move $\frac{1-x-y}{2}$ to rendezvous. This gives a competitive ratio of $\frac{1-y}{1-x-y}$. 
\item
$bd$ or $cd$ is non-faulty: then rendezvous of Case 1 occurs when $abc$ meet $d$ at the end. Therefore $b$ and $d$, or $c$ and $d$ move $\frac 12$ to rendezvous. If $bd$ is non-faulty then the optimal rendezvous time is $\frac{1-x}{2}$, which gives a competitive ratio of $\frac{1}{1-x}$. If $cd$ is non-faulty then the optimal rendezvous time is $\frac{y}{2}$, and thus the competitive ratio is $\frac{1}{y}$ in this case. 
\item
{\bf Competitive ratio :} Since $y\leq 1-x$ then $\frac{1}{1-x}\leq\frac{1}{y}$. So $f_1(x,y)=\max\{\frac{1-y}{1-x-y}, \frac{1}{y}\}$.
\end{itemize}
\end{itemize}

\vspace{.2cm}
\noindent{\bf Case 2 Rendezvous:} As depicted in Figure \ref{fig:case2}, robots $b$ and $c$ meet first, then $bc$ meets $a$ and at the end $abc$ and $d$ rendezvous. 
\begin{figure}[ht]
\begin{center}
\includegraphics[width=8cm]{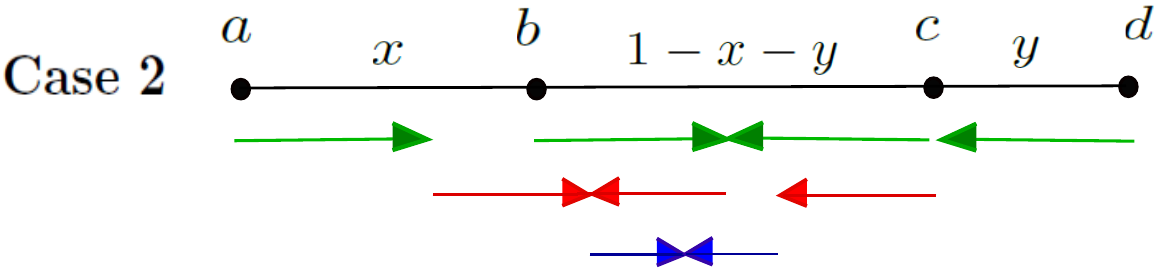}
\end{center}
\caption{Case 2 of Table \ref{tbl:cases}.}
\label{fig:case2}
\end{figure}

\begin{itemize}
\item[(1)]
{\bf The pairs of non-faulty robots with optimal and non-optimal rendezvous (columns 2 and 3 of Table~\ref{tbl:cases}):} If $bc$, $ac$ or $ad$ are non-faulty then the rendezvous of Case 2 is optimal. Therefore, if $bc$, $ac$ or $ad$ are non-faulty then the competitive ratio is one. For the cases when one of the pairs $ab$, $bd$, and $cd$ are non-faulty the rendezvous of Case 2 is not optimal. 
\item[(2)] 
{\bf The distance moved (column 4 of Table \ref{tbl:cases}):} 
\begin{itemize}
\item Rendezvous of $b$ and $c$: The distance between $b$ and $c$ is $1-x-y$. So $a$ and $b$ move to the right and $c$ and $d$ move to the left by a distance of $\frac{1-x-y}{2}$.
\item Rendezvous of $bc$ and $a$: The distance between $bc$ and $a$ is $x$. So $bc$ and $a$ move towards each other for $\frac{x}{2}$ to rendezvous, and $d$ moves $\frac x2$ to the left.
\item Rendezvous of $abc$ and $d$: The distance between $abc$ and $d$ is $y$. So  $abc$ and $d$ move towards each other for $\frac y2$.
\end{itemize}
\item[(3)]
{\bf The competitive ratio (column 5 of Table \ref{tbl:cases}):} let one of the pairs $ab$, $bd$, and $cd$ be non-faulty. 
\begin{itemize}
\item {\bf $ab$ is non-faulty:} then the rendezvous of Case 2 is complete when $a$ and $bc$ meet. Therefore the robots $a$ and $b$ move $\frac{1-x-y}{2}+\frac{x}{2}$ to rendezvous while in the optimal algorithm $a$ and $c$ move $\frac{x}{2}$ to rendezvous. This gives a competitive ratio of $\frac{1-y}{x}$.
\item {\bf $bd$ or $cd$ is non-faulty:} then the rendezvous of Case 2 is complete when $abc$ and $d$ meet at the end. Therefore the robots moved $\frac 12$. If $bd$ is non-faulty then the optimal rendezvous time is $\frac{1-x}{2}$, and thus the competitive ratio is $\frac{1}{1-x}$. If $cd$ non-faulty then the optimal rendezvous time is $\frac y2$, and thus the competitive ratio $\frac 1y$. Since $y\leq 1-x$ then $\frac{1}{1-x}\leq\frac 1y$.
\item {\bf Competitive ratio:} $f_2(x,y)=\max\{\frac{1-y}{x},\frac{1}{y}\}.$
\end{itemize}
\end{itemize}

\vspace{.2cm}
\noindent{\bf Case 3 Rendezvous:}  As depicted in Figure \ref{fig:case3} robots $b$ and $c$ meet first, then they split, $c$ moves to the right to meet $d$ and $b$ moves to the left to meet $a$. In the end $ab$ and $cd$ rendezvous. 
\begin{figure}[ht]
\begin{center}
\includegraphics[width=8cm]{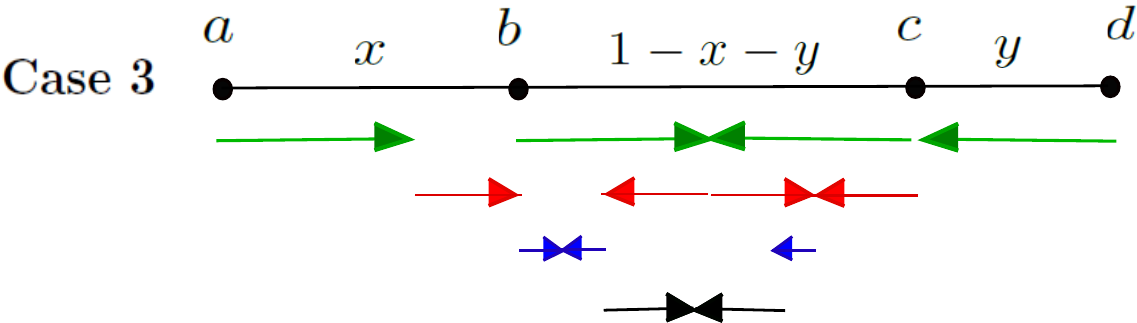}
\end{center}
\caption{Case 3 of Table \ref{tbl:cases}.}
\label{fig:case3}
\end{figure}
\begin{itemize}
\item[(1)]
{\bf The pairs of non-faulty robots with optimal and non-optimal rendezvous (columns 2 and 3 of Table \ref{tbl:cases}):}  It is easy to see that if $bc$ or $ad$ are non-faulty then the rendezvous of Case 3 is optimal. Therefore, if $bc$ or $ad$ are non-faulty then the competitive ratio is one. If one of the pairs $ab$, $ac$, $bc$, and $cd$ are non-faulty then their rendezvous is not optimal in Case 3. 
\item[(2)]
{\bf The distance moved (column 4 of Table \ref{tbl:cases}):} 
\begin{itemize}
\item 
Meeting of $b$ and $c$: The distance between $b$ and $c$ is $1-x-y$. So $b$ and $c$ move $\frac{1-x-y}{2}$ towards each other. Also $a$ move to the right and $d$ move to the left by a distance of $\frac{1-x-y}{2}$. 
\item
Rendezvous of $c$ and $d$: The distance between $c$ and $d$ is $y$. So $c$ and $d$ move towards each other for $\frac{y}{2}$ to meet. Also $a$ and $b$ move towards each other for $\frac y2$. Since the distance between $a$ and $b$ is $x$ and $y\leq x$ then $a$ and $b$ do not meet at the time of rendezvous of $cd$.
\item
Rendezvous of $a$ and $b$: When the rendezvous $ad$ occurs the distance between $a$ and $b$ is $x-y$. Therefore $a$ and $b$ continue moving towards each other for $\frac{x-y}{2}$ while $cd$ moves $\frac{x-y}{2}$ to the left to meet $ab$.
\item
Rendezvous of $ab$ and $cd$: When $a$ and $b$ meet the distance between $ab$ and $cd$ is $1-(1-x-y)-x$. So for the final rendezvous $ab$ and $cd$ move $\frac{y}{2}$ towards each other.  
\end{itemize}
\item[(3)]
{\bf The competitive ratio (column 5 of Table \ref{tbl:cases}):} let one of the pairs $ac$, $bd$, $ab$ and $cd$ be non-faulty. 
\begin{itemize}
\item
$ac,$ or $bd$ is non-faulty: then the rendezvous of Case 3 is complete when $ab$ and $cd$ meet at the end. Therefore the robots $a$ and $c$, or $b$ and $d$ move $\frac{1}{2}$ to rendezvous. If $ac$ is non-faulty then the optimal rendezvous time of $a$ and $c$ is $\frac{1-y}{2}$, and thus the competitive ratio is $\frac{1}{1-y}$. If $bd$ is non-faulty then the optimal rendezvous time of $b$ and $d$ is $\frac{1-x}{2}$, and thus the competitive ratio is $\frac{1}{1-x}$. 
\item
$ab$ is non-faulty: then the rendezvous is complete when $b$ and $a$ meet after the meeting of $b$ and $c$. So $a$ and $b$ move $\frac{1-x-y}{2}+\frac{x}{2}$ to rendezvous while an optimal rendezvous requires $\frac{x}{2}$. So the competitive ratio is $\frac{1-y}{x}$.
\item
$cd$ is non-faulty then rendezvous of Case 3 completes when $c$ and $d$ meet after the meeting of $b$ and $c$. So $c$ and $d$ move $\frac{1-x-y}{2}+\frac{y}{2}$ to rendezvous while an optimal rendezvous requires $\frac{y}{2}$. So the competitive ratio is $\frac{1-x}{y}$.   
\item
{\bf Competitive ratio:}
Since $y\leq x$ then $\frac{1}{1-y}\leq \frac{1}{1-x}$ and $\frac{1-y}{x}\leq \frac{1-x}{y}$. So $f_3(x,y)=\max\{\frac{1-x}{y}, \frac{1}{1-x}\}$.  
\end{itemize}
\end{itemize}     

\vspace{.2cm}
\noindent{\bf Case 4 Rendezvous:} As depicted in Figure \ref{fig:case4}, robots $b$ and $c$ meet first, then $bc$ meets $d$ and at the end $bcd$ and $a$ rendezvous. 
\begin{figure}[ht]
\begin{center}
\includegraphics[width=8cm]{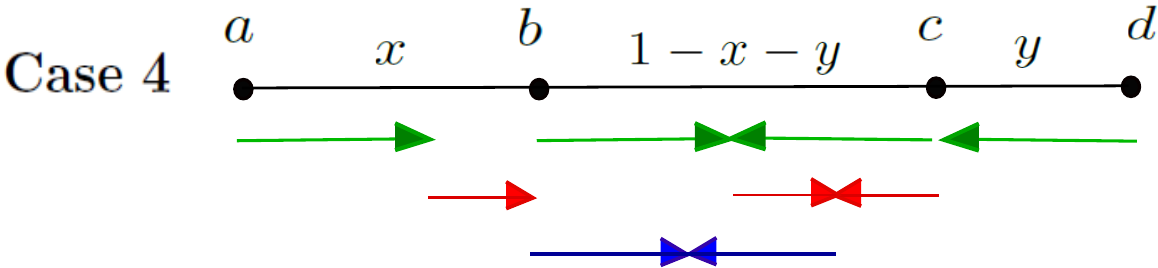}
\end{center}
\caption{Case 4 of Table \ref{tbl:cases}.}
\label{fig:case4}
\end{figure}

\begin{itemize}
\item[(1)]
{\bf The pairs of non-faulty robots with optimal and non-optimal rendezvous (columns 2 and 3 of Table~\ref{tbl:cases}):} If $bc$, $bd$ or $ad$ are non-faulty then the rendezvous of Case 4 is optimal. Therefore, if $bc$, $bd$ or $ad$ are non-faulty then the competitive ratio is one. For the cases when one of the pairs $ab$, $ac$, and $cd$ are non-faulty the rendezvous of Case 4 is not optimal. 
\item[(2)] 
{\bf The distance moved (column 4 of Table \ref{tbl:cases}):} 
\begin{itemize}
\item Rendezvous of $b$ and $c$: The distance between $b$ and $c$ is $1-x-y$. So $a$ and $b$ move to the right and $c$ and $d$ move to the left by a distance of $\frac{1-x-y}{2}$.
\item Rendezvous of $bc$ and $d$: The distance between $bc$ and $d$ is $y$. So $bc$ and $d$ move towards each other for $\frac{y}{2}$ to rendezvous, and $a$ moves $\frac y2$ to the right.
\item Rendezvous of $bcd$ and $a$: The distance between $bcd$ and $a$ is $x$. So  $bcd$ and $a$ move towards each other for $\frac x2$.
\end{itemize}
\item[(3)]
{\bf The competitive ratio (column 5 of Table \ref{tbl:cases}):} let one of the pairs $ab$, $ac$, and $cd$ be non-faulty. 
\begin{itemize}
\item {\bf $ab$ or $ac$ is non-faulty:} then the rendezvous of Case 4 is complete when $a$ and $bcd$ meet at the end. Therefore the robots $a$ and $c$ or $a$ and $c$ move $\frac 12$ to rendezvous. If $ab$ is non-faulty then the optimal rendezvous time is $\frac{x}{2}$, and thus the competitive ratio is $\frac{1}{x}$. If $ac$ non-faulty then the optimal rendezvous time is $\frac {1-y}{2}$, and thus the competitive ratio $\frac 1x$. Since $x\leq 1-y$ then $\frac{1}{1-y}\leq\frac 1x$.
\item {\bf $cd$ is non-faulty:} then the rendezvous of Case 4 is complete when $cb$ and $d$ meet. Therefore the robots $c$ and $d$ move $\frac{1-x-y}{2}+\frac{y}{2}$ to rendezvous while in the optimal algorithm $c$ and $d$ move $\frac{y}{2}$ to rendezvous. This gives a competitive ratio of $\frac{1-x}{y}$.
\item {\bf Competitive ratio:} $f_4(x,y)=\max\{\frac{1-x}{y},\frac{1}{x}\}.$
\end{itemize}
\end{itemize}

\vspace{.2cm}
\noindent{\bf Case 5 Rendezvous:} As depicted in Figure \ref{fig:case5}, robots $c$ and $d$ meet first, then $a$ and $b$ meet, and finally $ab$ meets $cd$. 
\begin{figure}[ht]
\begin{center}
\includegraphics[width=8cm]{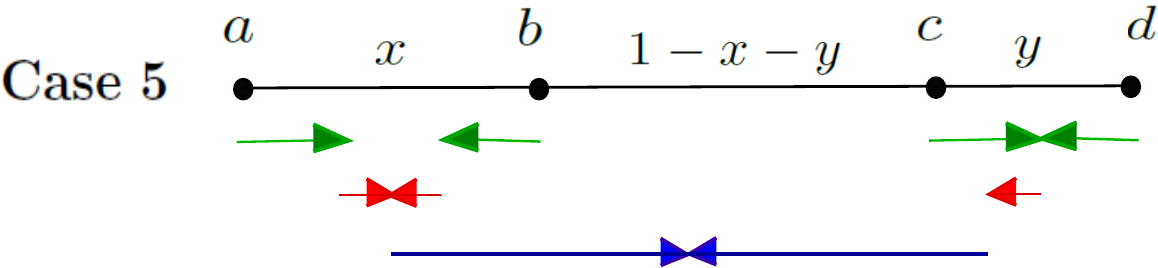}
\end{center}
\caption{Case 5 of Table \ref{tbl:cases}.}
\label{fig:case5}
\end{figure}

\begin{itemize}
\item[(1)]
{\bf The pairs of non-faulty robots with optimal and non-optimal rendezvous (columns 2 and 3 of Table~\ref{tbl:cases}):} If $ab$, $cd$ or $ad$ are non-faulty then the rendezvous of Case 5 is optimal. Therefore, if $ab$, $cd$ or $ad$ are non-faulty then the competitive ratio is one. For the cases when one of the pairs $ac$, $bc$, and $bd$ are non-faulty the rendezvous of Case 5 is not optimal. 
\item[(2)] 
{\bf The distance moved (column 4 of Table \ref{tbl:cases}):} 
\begin{itemize}
\item Rendezvous of $c$ and $d$: The distance between $c$ and $d$ is $y$. So $c$ and $d$ move towards each other for $\frac y2$. Also, $a$ and $b$ move towards each other for $\frac{y}{2}$. Since the distance between $a$ and $b$ is $x$ and $x\leq y$ they do not rendezvous by time $\frac y2$.
\item Rendezvous of $a$ and $b$: The distance between $a$ and $b$ is $x-y$ at the time of rendezvous of $cd$. So $a$ and $b$ move towards each other for another $\frac{x-y}{2}$ to rendezvous. Also, $cd$ moves $\frac{x-y}{2}$ to the left to rendezvous with $ab$.
\item Rendezvous of $ab$ and $cd$: The distance between $ab$ and $cd$ is $1-x$ at the time of rendezvous of $ab$. So $ab$ and $cd$ move towards each other for $\frac{1-x}{2}$ to rendezvous.
\end{itemize}
\item[(3)]
{\bf The competitive ratio (column 5 of Table \ref{tbl:cases}):} let one of the pairs $ac$, $bc$, and $bd$ be non-faulty. 
\begin{itemize}
\item {\bf $ac$, $bc$, or $bd$ is non-faulty:} then the rendezvous of Case 5 is complete when $ab$ and $cd$ meet at the end. Therefore for all the cases the robots move $\frac 12$ to rendezvous. If $ac$ is non-faulty then the optimal rendezvous time is $\frac{1-y}{2}$, and thus the competitive ratio is $\frac{1}{1-y}$. If $bc$ non-faulty then the optimal rendezvous time is $\frac {1-x-y}{2}$, and thus the competitive ratio $\frac{1}{1-x-y}$. If $bd$ is non-faulty the the optimal rendezvous time is $\frac{1-x}{2}$, and thus the competitive ratio is $\frac{1}{1-x}$. Since $1-x-y\leq \max\{1-x, 1-y\}$ then $\max\{\frac{1}{1-x},\frac{1}{1-y}\}\leq\frac {1}{1-x-y}$.
\item {\bf Competitive ratio:} $f_5(x,y)=\frac{1}{1-x-y}.$
\end{itemize}
\end{itemize}

\vspace{.2cm}
\noindent{\bf Case 6 Rendezvous:} As depicted in Figure \ref{fig:case6}, robots $c$ and $d$ meet first, then $cd$ meets $b$ and at the end $bcd$ and $a$ rendezvous. 
\begin{figure}[ht]
\begin{center}
\includegraphics[width=8cm]{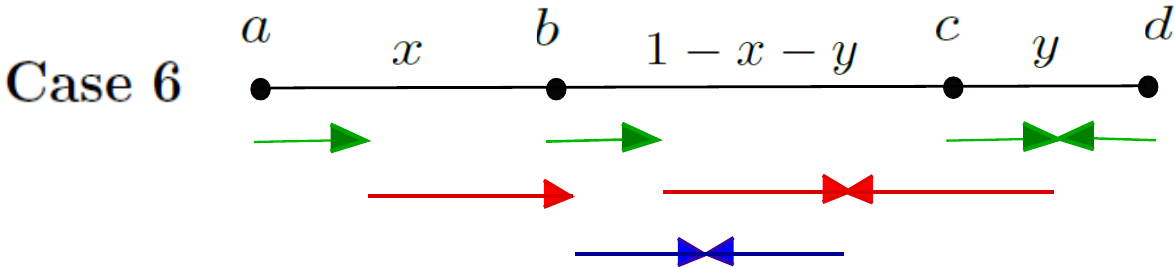}
\end{center}
\caption{Case 6 of Table \ref{tbl:cases}.}
\label{fig:case6}
\end{figure}
\begin{itemize}
\item[(1)]
{\bf The pairs of non-faulty robots with optimal and non-optimal rendezvous (columns 2 and 3 of Table~\ref{tbl:cases}):} If $cd$, $bd$ or $ad$ are non-faulty then the rendezvous of Case 4 is optimal. Therefore, if $cd$, $bd$ or $ad$ are non-faulty then the competitive ratio is one. For the cases when one of the pairs $ab$, $ac$, and $bc$ are non-faulty the rendezvous of Case 6 is not optimal. 
\item[(2)] 
{\bf The distance moved (column 4 of Table \ref{tbl:cases}):} 
\begin{itemize}
\item Rendezvous of $c$ and $d$: The distance between $c$ and $d$ is $y$. So $c$ and $d$ move towards each other for $\frac y2$. Also, $a$ and $b$ move to the right by a distance of $\frac{y}{2}$.
\item Rendezvous of $cd$ and $b$: The distance between $cd$ and $b$ is $1-x-y$. So $cd$ and $b$ move towards each other for $\frac{1-x-y}{2}$ to rendezvous, and $a$ moves $\frac{1-x-y}{2}$ to the right.
\item Rendezvous of $bcd$ and $a$: The distance between $bcd$ and $a$ is $x$. So  $bcd$ and $a$ move towards each other for $\frac x2$.
\end{itemize}
\item[(3)]
{\bf The competitive ratio (column 5 of Table \ref{tbl:cases}):} let one of the pairs $ab$, $ac$, and $bc$ be non-faulty. 
\begin{itemize}
\item {\bf $ab$ or $ac$ is non-faulty:} then the rendezvous of Case 6 is complete when $a$ and $bcd$ meet at the end. Therefore the robots $a$ and $b$ or $a$ and $c$ move $\frac 12$ to rendezvous. If $ab$ is non-faulty then the optimal rendezvous time is $\frac{x}{2}$, and thus the competitive ratio is $\frac{1}{x}$. If $ac$ non-faulty then the optimal rendezvous time is $\frac {1-y}{2}$, and thus the competitive ratio $\frac 1x$. Since $x\leq 1-y$ then $\frac{1}{1-y}\leq\frac 1x$.
\item {\bf $bc$ is non-faulty:} then the rendezvous of Case 6 is complete when $cd$ and $b$ meet. Therefore the robots $b$ and $c$ move $\frac y2+\frac{1-x-y}{2}$ to rendezvous while in the optimal algorithm $b$ and $c$ move $\frac{1-x-y}{2}$ to rendezvous. This gives a competitive ratio of $\frac{1-x}{1-x-y}$.
\item {\bf Competitive ratio:} $f_6(x,y)=\max\{\frac{1-x}{1-x-y},\frac{1}{x}\}.$
\end{itemize}
\end{itemize}
Now that we have all possible rendezvous algorithms and their corresponding competitive ratios we can prove that $1+\phi$ is a lower bound of the competitive ratio of any rendezvous algorithm for four robots two of which are faulty.
\begin{lemma}
The competitive ratio of any rendezvous algorithm for four robots two of which are faulty is bounded below by $1+\phi$.
\end{lemma}
\begin{proof}
Consider the point $(x,y)=(\frac{1}{1+\phi}, \frac{1}{\phi(1+\phi)})$. Inserting the value of $x$ and $y$ into functions $f_i$, $1\leq i\leq 6$ imply that 
$$f_i \left(\frac{1}{1+\phi}, \frac{1}{\phi(1+\phi)} \right)=1+\phi.$$
This proves that for any rendezvous algorithm the competitive ratio is at least $1+\phi$.
\qed
\end{proof} 

This completes the proof of Theorem \ref{thm:4-robots-thm}.

\section{Proof of Theorem~\ref{thm:4-robots-UB}}

In what follows we aim to prove that $1+\phi$ is also an upper bound. We present a rendezvous algorithm whose competitive ratio is at most $1+\phi$. Therefore the algorithm to be introduced later is an optimal  algorithm.

The function  to be optimized  represents the rendezvous time of any algorithm and the resulting optimization problem is defined as follows:
\begin{equation}
\label{opt:eq}
\begin{array}{ll}
& \max_{x, y} \min_{1 \leq i \leq 6} \left\{ f_i (x,y) \right\} \\
\mbox{\bf Subject to:}~~
& 0 \leq y \leq x \leq 1 \mbox{ and } x+y \leq 1 .
\end{array}
\end{equation}

Now we analyze the rendezvous time of all possible algorithms concerning four robots two of which are faulty. 
It is easy to prove that $f_6$ is dominated by $f_1$ and $f_4$ is dominated by $f_2$ in the sense that $f_6 (x,y) \leq f_1 (x,y)$ and $f_4 (x,y) \leq f_2 (x,y)$, for all $x,y$. 
\begin{lemma}
\label{lm1}
$f_6 (x,y) \leq f_1 (x, y)$, for all $x, y \in [0,1]$ such that $y \leq \min \{ x , 1-x\}$.
\end{lemma}
\begin{proof} (Lemma~\ref{lm1})
Recall from the definitions that 
$$
f_1 (x, y) = \max \left\{ \frac{1-y}{1-x-y}, \frac{1}{y} \right\}
\mbox{ and }
f_6 (x, y) =\max \left\{ \frac{1}{x}, \frac{1-x}{1-x-y} \right\}.
$$ 
Hence, the inequality $f_6 (x,y) \leq f_1 (x, y)$ is immediate since $y \leq x$. This proves Lemma~\ref{lm1}.
\qed
\end{proof}

\begin{lemma}
\label{lm2}
$f_4 (x,y) \leq f_2 (x, y)$, for all $x, y \in [0,1]$ such that $y \leq \min \{ x , 1-x\}$.
\end{lemma}
\begin{proof} (Lemma~\ref{lm2})
Recall from the definitions that 
$$
f_2 (x, y) = \max \left\{ \frac{1-y}{x}, \frac{1}{y} \right\}
\mbox{ and }
f_4 (x, y) =\max \left\{ \frac{1}{x}, \frac{1-x}{y} \right\}.
$$ Since $y \leq x$ we have $\frac 1x \leq \frac 1y$ and $\frac{1-x}{y}  \leq \frac{1}{y}$. Hence, the inequality $f_4 (x,y) \leq f_2 (x, y)$ is immediate since $y \leq x$. This proves (Lemma~\ref{lm2}). \qed
\end{proof}

Therefore the Optimization Problem~\eqref{opt:eq} considered above is equivalent to the following:
\begin{equation}
\label{opt1:eq}
\begin{array}{ll}
& \max_{x, y} \min_{3 \leq i \leq 6} \left\{ f_i (x,y) \right\} \\
\mbox{\bf Subject to:}~~
& 0 \leq y \leq x \leq 1 \mbox{ and } x+y \leq 1 .
\end{array}
\end{equation}
Further, we may assume without loss of generality that $0 \leq x < 1$. Next we detect the regions of $\mathbb{R}^2$ for which  the rendezvous algorithms of Cases 3-6 have competitive ratio of at most $1+\phi$. Indeed Lemmas \ref{lm3}-\ref{lm6} obtain the region $R(f_i)$, $3\leq i\leq 6$, in which the competitive ratio corresponding to Case $i$, $f_i(x,y)$, is at most $1+\phi$. 

\begin{lemma}
\label{lm3}
Let $R(f_3)=\{(x,y)| x\leq \frac{\phi}{1+\phi}, y\geq \frac{1}{1+\phi}-\frac{x}{1+\phi}\}$.
Then $f_3(x,y)\leq 1+\phi$ if and only if $(x,y)$ is a point in $R(f_3)$. 
\end{lemma}
\begin{proof} (Lemma~\ref{lm3})
Recall that $f_3(x,y)=\max\{\frac{1}{x-1},\frac{1-x}{y}\}$. Therefore $f_3(x,y)\leq 1+\phi$ if and only if $\frac{1}{x-1}\leq 1+\phi$ and $\frac{1-x}{y}\leq 1+\phi$. 
Using elementary calculations we can show easily that
\begin{align*}
\frac{1}{1-x}  \leq 1+\phi
&\Leftrightarrow x   \leq  \frac{\phi}{1+\phi} \\
\frac{1-x}{y}   \leq 1+\phi
&\Leftrightarrow y\geq \frac{1}{1+\phi}-\frac{x}{1+\phi}\\
\end{align*}
This proves the lemma.
\qed
\end{proof}


\begin{lemma}
\label{lm4}
Let $R(f_4)=\{(x,y)| x\geq \frac{1}{1+\phi}, y\geq \frac{1}{1+\phi}-\frac{x}{1+\phi}\}$.
Then $f_4(x,y)\leq 1+\phi$ if and only if $(x,y)$ is a point in $R(f_4)$. 
\end{lemma}
\begin{proof} (Lemma~\ref{lm4})
Recall that $f_4(x,y)=\max\{\frac{1}{x},\frac{1-x}{y}\}$. Therefore $f_4(x,y)\leq 1+\phi$ if and only if $\frac{1}{x}\leq 1+\phi$ and $\frac{1-x}{y}\leq 1+\phi$. 
Using elementary calculations we can show easily that
\begin{align*}
\frac{1}{x}  \leq 1+\phi
&\Leftrightarrow x   \geq  \frac{1}{1+\phi} \\
\frac{1-x}{y}   \leq 1+\phi
&\Leftrightarrow y\geq \frac{1}{1+\phi}-\frac{x}{1+\phi}\\
\end{align*}
This proves the lemma.
\qed
\end{proof}

\begin{lemma}
\label{lm5}
Let $R(f_5)=\{(x,y)| y\leq \frac{\phi}{1+\phi}-x\}$.
Then $f_5(x,y)\leq 1+\phi$ if and only if $(x,y)$ is a point in $R(f_5)$. 
\end{lemma}
\begin{proof} (Lemma~\ref{lm5})
Recall that $f_5(x,y)=\frac{1}{1-x-y}$. Therefore 
\begin{align*}
\frac{1}{1-x-y}  \leq 1+\phi
&\Leftrightarrow 0   \leq  \phi-(1+\phi)x-(1+\phi)y\\
&\Leftrightarrow y\leq \frac{\phi}{1+\phi}-x\\
\end{align*}
This proves the lemma.
\qed
\end{proof}

\begin{lemma}
\label{lm6}
Let $R(f_6)=\{(x,y)| x\geq \frac{1}{1+\phi}, y\leq \frac{\phi}{1+\phi}-\frac{\phi x}{1+\phi}\}$.
Then $f_6(x,y)\leq 1+\phi$ if and only if $(x,y)$ is a point in $R(f_6)$. 
\end{lemma}
\begin{proof} (Lemma~\ref{lm6})
Recall that $f_6(x,y)=\max\{\frac{1}{x},\frac{1-x}{1-x-y}\}$. Therefore $f_6(x,y)\leq 1+\phi$ if and only if $\frac{1}{x}\leq 1+\phi$ and $\frac{1-x}{1-x-y}\leq 1+\phi$. 
Using elementary calculations we can show easily that
\begin{align*}
\frac{1}{x}  \leq 1+\phi
&\Leftrightarrow x   \geq  \frac{1}{1+\phi} \\
\frac{1-x}{1-x-y}   \leq 1+\phi
&\Leftrightarrow 0\leq \phi-(1+\phi)y-\phi x
\Leftrightarrow y\leq \frac{\phi}{1+\phi}-\frac{\phi x}{1+\phi}\\
\end{align*}
This proves the lemma.
\qed
\end{proof}

We now show that $\bigcup_{3\leq i\leq 6} R(f_i)$ covers the area $\{(x,y)| 0 \leq y \leq x \leq 1, x+y \leq 1\}$, which is the region where our optimal problem \ref{opt:eq} is defined.
\begin{lemma}
\label{lem:cover}
For any $(x,y)$ such that $0 \leq y \leq x \leq 1$ and $x+y \leq 1$, we have that $(x,y)\in \bigcup_{3\leq i\leq 6} R(f_i)$.
\end{lemma}
\begin{proof}
(Lemma \ref{lem:cover})
For any $0\leq x\leq 1$, either $x\leq \frac{1}{1+\phi}$ or $x\geq \frac{1}{1+\phi}$. First let $x\leq \frac{1}{1+\phi}$. We prove that for $x\leq \frac{1}{1+\phi}$ we have the following.
\begin{equation}
\label{eq:lem-cover}
\frac{1}{1+\phi}-\frac{x}{1+\phi}\leq \frac{\phi}{1+\phi}-x.
\end{equation}
Note that the left term of Inequality \ref{eq:lem-cover} is the border of $R(f_3)$, and the right term is the border of $R(f_5)$. We have
\begin{align*}
\frac{1}{1+\phi}-\frac{x}{1+\phi}\leq \frac{\phi}{1+\phi}-x
&\Leftrightarrow 1-x\leq \phi-x-\phi x\\
&\Leftrightarrow x\leq \frac{\phi-1}{\phi}=\frac{1}{1+\phi}
\end{align*}
This proves Inequality \ref{eq:lem-cover}. If $y\leq \frac{1}{1+\phi}-\frac{x}{1+\phi}$ then by inequality (5) we have that $y\leq \frac{\phi}{1+\phi}-x$, and thus $(x,y)\in R(f_5)$.
Now let $y\geq \frac{1}{1+\phi}-\frac{x}{1+\phi}$. Since $x\leq\frac{1}{1+\phi}<\frac{\phi}{1+\phi}$ then by definition of $R(f_3)$ we have that $(x,y)\in R(f_3)$. 

We now consider the case where $x\geq \frac{1}{1+\phi}$. If $y\leq \frac{\phi}{1+\phi}-\frac{\phi x}{1+\phi}$ then by definition of $R(f_6)$ we have that $(x,y)\in R(f_6)$. Now let $y\geq \frac{\phi}{1+\phi}-\frac{\phi x}{1+\phi}$. Then $y\geq \frac{1}{1+\phi}-\frac{x}{1+\phi}$, and so by definition of $R(f_4)$ we have that $(x,y)\in R(f_4)$.
This proves that if the point $(x,y)$ is such that $0\leq y\leq x\leq 1$ and $x+y\leq 1$ then $(x,y)$ belongs at least one $R(f_i)$, $3\leq i\leq 6$. 
\qed
\end{proof}
%

After discussing all the previous details, we are now in a position to prove the upper bound of $1+\phi$ for four robots two of which are faulty.

\begin{proof} (Theorem \ref{thm:4-robots-UB})
We state the following algorithm with competitive ratio at most $1+\phi$ for the rendezvous of four robots.

\vspace{-0.3cm}
\begin{algorithm}[H]
\caption{FRR (Four Robots Rendezvous):}
\label{alg:frr}
\begin{algorithmic}[1]
\State{The robots broadcast their coordinates.}
\State{Compute the distance between the first two robots, $x$, and the distance between the last two robots, $y$.}
\If{$(x,y)\in R(f_3)$}
\State{execute Rendezvous of Case 3.}
\EndIf
\If{$(x,y)\in R(f_4)\setminus R(f_3)$}
\State{execute Rendezvous of Case 4.}
\EndIf
\If{$(x,y)\in R(f_5)\setminus \bigcup_{3\leq i\leq 4}R(f_i)$}
\State{execute Rendezvous of Case 5.}
\EndIf
\If{$(x,y)\in R(f_6)\setminus \bigcup_{3\leq i\leq 5}R(f_i)$}
\State{execute Rendezvous of Case 6.}
\EndIf
\end{algorithmic}
\end{algorithm}
By Lemmas \ref{lm1}-\ref{lem:cover}, we can conclude that Algorithm \ref{alg:frr} has a competitive ratio of at most $1+\phi$.
\qed
\end{proof}

\end{document}